%% file: JLBB.tex
\documentclass{article}

\usepackage{amssymb}
\usepackage{amsmath}
\usepackage[dvips]{graphicx}
\usepackage{subfigure}

\newcommand{\s}{{\mathbf s}}
\renewcommand{\t}{{\mathbf t}}

\renewcommand{\v}{{\mathbf v}}
\newcommand{\x}{{\mathbf x}}
\newcommand{\e}{\mathbf{e}}
\newcommand{\area}{\mathrm{area}}
\newcommand{\vol}{\mathrm{vol}}
\newcommand{\sh}{\mathrm{sh}}
\usepackage{latexsym}
\usepackage{graphics}
\newtheorem{lemma}{Lemma}
\newtheorem{theorem}{Theorem}
\newtheorem{definition}{Definition}
\newtheorem{corollary}{Corollary}
\newtheorem{fact}{Fact}
\newtheorem{conjecture}{Conjecture}
\newenvironment{proof}
{\noindent {\bf Proof.}}
{$\blacksquare$}
\begin{document}

\title{Euclidean versus hyperbolic congestion in idealized versus experimental networks}

\author{Edmond Jonckheere and Mingji Lou \& Francis Bonahon \\
Department of Electrical Engineering \& Department of Mathematics\\
University of Southern California \\
Los Angeles, CA 90089\\
{\tt jonckhee@usc.edu} and {\tt mingjilou@gmail.com} \& {\tt fbonahon@math.usc.edu}\\
\\
Yuliy Baryshnikov\\
Room MH 2C-361\\
Bell Laboratories\\
600 Mountain Ave\\
Murray Hill, NJ 07974-0636\\
{\tt ymb@research.bell-labs.com}\\
}
\maketitle
\noindent 

\noindent 

\noindent \textbf{\textit{Abstract: }}This paper proposes a mathematical justification of the phenomenon of extreme congestion at a very limited number of nodes in very large networks. It is argued that this phenomenon occurs as a combination of the negative curvature property of the network together with minimum length routing. More specifically, it is shown that, 
in a large $n$-dimensional hyperbolic ball $B$ of radius $R$ viewed as a roughly similar model of a Gromov hyperbolic network, 
the proportion of traffic paths transiting through a small ball near the center is $\Theta(1)$, 
whereas, in a Euclidean ball, the same proportion scales as $\Theta\left(\frac{1}{R^{n-1}}\right)$. 
This discrepancy persists for the traffic load, 
which at the center of the hyperbolic ball scales as $\vol^{2}(B)$, 
whereas the same traffic load scales as $\vol^{1+\frac{1}{n}}(B)$ in the Euclidean ball. 
This provides a theoretical justification of the experimental exponent discrepancy observed by Narayan and Saniee 
between traffic loads in Gromov-hyperbolic networks from the Rocketfuel data base and synthetic Euclidean lattice networks. 
It is further conjectured that for networks that do not enjoy the obvious symmetry of hyperbolic and Euclidean balls, 
the point of maximum traffic is near the center of mass of the network.

\noindent 
\section{Introduction}

One of the most important challenges in large and wide area networks is to overcome the traffic congestion problem. The queuing feature at routers can create a logic bottleneck between two users. Correspondingly, insufficient bandwidth on the physical links between routers is a contributor to congestion. The current congestion control technologies in communication networks are based on the feedback from the congested node to the source to slow down the packet flow rate, such as bidirectional congestion control and Random Early Detection (RED). However, these technologies can only be applied once the congestion has happened to some degree, and it is only based on the \textit{local} point of view of some queue overflow along the source-to-target path. 

From the large scale point of view, it has been experimentally observed that, on the Internet and other networks, traffic seems to concentrate quite heavily on some very small subsets. The major result of this paper is to show that the deeper reason behind this congestion \textit{in the large scale} is the combination of the least cost routing together with the negative curvature of the network. 

Roughly speaking, a network is negatively curved if its graph can be approximated by a negatively curved Riemannian manifold. In practice, however, it has become customary to check the Gromov   \textit{Thin Triangle Condition (TTC)}, meaning that the least cost paths between three vertices ``arch inside'' the triangle, giving it a ``thin'' external appearance. The connection between negatively curved surfaces and the TTC can be understood on the basis that a triangle drawn on a negatively curved surface has the sum of its internal angles $<\pi $, giving it a ``thin'' appearance. The formalization of this equivalence is, however, must harder and is known under the Bonk-Schramm theorem~\cite{bonk_schramm}. 

Over the past few years, there has been mounting evidence that many (wired and wireless) communication networks are negatively curved~\cite{Matt_thesis,scaled_gromov,JonckheereLohsoonthornACC2004,JonckheereLohsoonthornMED2002}. Hierarchical networks have been found to have ``hidden hyperbolic structure''~\cite{dima_hidden_hyperbolic}. In view of the above graph-manifold identification, we can certainly argue on a negatively curved $n$-dimensional Riemannian manifold. To simplify the exposition in this Introduction, we invoke the nontrivial fact, proved in the Appendix, that congestion in very large hyperbolic ball models does not depend on the dimension. Therefore, it is legitimate to argue on a negatively curved surface, as we will do in the remainder of this Introduction.  

A well-known feature of least length paths on a negatively curved surface is the fact that two geodesics starting from points within an arbitrarily small disk eventually diverge exponentially.  
In this paper, we somehow reverse that argument and show that least length paths departing at remote points of a convex subset $X$ of the surface ``converge'' to a single point where the ``density of geodesics'' is the highest. The latter is the point of maximum congestion. We conjecture that this congestion point is near the {\it center of mass} of the convex subset, and {\it prove} that the maximum congestion point and the center of mass coincide for very large disks embedded in the surface. 

Naturally, there might be some questions as to the relevance of this idealized analysis to real networks. In fact, some recent experimental results~\cite{arXiv_dmitri} have confirmed the validity of our analysis. The network congestion metric adopted in~\cite{arXiv_dmitri} is the betweennness centrality~\cite{congestion_tree}. The continuous geometry traffic congestion metric adopted here is the average length of the geodesics in a small convex subset $X$. We call the later {\it load measure}, $\Lambda_t (X)$, because, if $X$ is thought as a sub-network, $\Lambda_t (X)$ is meant to be the number of packets in that sub-network. Consistently with the major theme of this paper, we prove that, if $X$ is a disk of fixed radius, $\Lambda_t (X)$ is maximum when the disk has its center at the center of mass. But more importantly, we show that, in a large disk $B_R(0)$ of hyperbolic surface, $\max_{X\subseteq B_R(0),\area(X)=1} \Lambda_t (X)$ scales as $(\mbox{area}(B_R(0)))^{2}$. On the other hand, it was shown in~\cite{arXiv_dmitri} that, for networks for the Rocketfuel data base~\cite{rocketfuel}, the maximum betweennness centrality scales as $N^{2}$, where $N$ is the number of vertices. Since $N$ can be interpreted as the area, this confirms the validity of our theoretical model. 

An outline of the paper follows: 
In Section~\ref{s:basic_def}, 
we set the basic traffic metrics in graph models. 
In Section~\ref{s:basic_conjectures}, 
a simple concept of negatively curved planar graphs is introduced and 
the general facts about congestion in negatively curved graphs are proposed as conjectures.  
Immediately thereafter, in Section~\ref{s:examples}, 
we show that our conjectures hold 
in selected planar graph examples. In Section~\ref{s:graphs_to_manifolds}, 
we propose the more general concept of Gromov-hyperbolic graphs and 
we invoke the Bonk-Schramm theorem to argue that 
traffic in Gromov-hyperbolic graphs can be analyzed on negatively curved Riemannian manifold models. 
In Section~\ref{s:dif_geom}, we formulate our major results dealing with traffic in  
subsets of $\mathbb{H}^n$. 
Probably the most significant results are Theorems~\ref{t:euclidean_traffic} and~\ref{t:hyperbolic_traffic}, 
{\it proving} that traffic in a small subset $X$ of large hyperbolic and Euclidean balls $B_R(0)$ scale 
as $\vol(B_R(0))^2$ and $\vol(B_R(0))^{1.5}$, respectively.  
(The proofs are relegated to the Appendix, though.)  
Theorem~\ref{t:min_inertia}
and~\ref{t:max_traffic} formulate the basic minimum inertia versus maximum traffic issues, 
{\it proving} the conjectures for hyperbolic balls $X$.  
Finally, in Section~\ref{s:real_networks}, we show that no matter how theoretical our models are, 
they turn out to be surprisingly accurate at predicting asymptotic traffic distribution 
in realistic networks.

\noindent 
\section{Traffic metrics}
\label{s:basic_def}

Let $G=(V,E)$ be a (connected) graph specified by its vertex set $V$ and its edge set $E$ and endowed with a (symmetric) distance function $d:G\times G\to \mathbb{R}^{+} $. A \textit{path} $p(s,t)$ from $s$ to $t$ is a continuous map $[0,l]\to G$ such that $p(s,t)(0)=s$ and $p(s,t)(l)=t$. The \textit{weight} of an edge $e=xy$ is defined as $w(e)=d(x,y)$. The \textit{length} of the path is defined as $\ell (p(s,t))=\sum _{e\subseteq p(s,t)}w(e) $. A \textit{geodesic} $[s,t]$ is a path such that $\ell \left(\left[s,t\right]\right)\le \ell \left(p\left(s,t\right)\right)$, for all paths $p(s,t)$ joining $s$ to $t$. 

The traffic on the graph is driven by a \textit{demand measure }$\Lambda _{d} :V\times V\to \mathbb{R}^{+}$, where the demand $\Lambda _{d} (s,t)$ is the traffic rate (e.g., number of packets per second) to be transmitted from the source $s$ to the destination target $t$. Assume that the routing protocol sends the packets from the source $s$ to the target $t$ along the geodesic $\left[s,t\right]$ with probability $\pi \left(\left[s,t\right]\right)$. It is indeed customary as a load balancing strategy to randomize the Dijkstra algorithm so as to distribute the traffic more evenly~\cite{mingjithesis,Cisco05}. Under this scheme, the geodesic $\left[s,t\right]$ inherits a traffic rate measure $\tau \left(\left[s,t\right]\right)=\Lambda _{d} (s,t)\pi \left(\left[s,t\right]\right)$. An edge $e$ laying on the path $\left[s,t\right]$ inherits from that path a traffic $\tau \left(\left[s,t\right]\right)$. Aggregating this traffic over all source-target pairs and all geodesics traversing the edge $e$ yields the traffic rate sustained by the edge $e$, 
\[\tau (e)=\sum _{(s,t)\in V\times V}\sum _{\left[s,t\right]\supseteq e}\tau \left(\left[s,t\right]\right)=  \sum _{(s,t)\in V\times V}\sum _{\left[s,t\right]\supseteq e}\Lambda _{d} (s,t)\pi \left(\left[s,t\right]\right)  . \] 
\noindent Observe that $\tau (e)\ell (e)$ can be interpreted as the {\it traffic load}, that is the number of packets in the edge $e$. 

This paper is essentially concerned with existence of a sub-network that has extremely high traffic load. Formally, given a connected subgraph $X\subseteq G$, containing some edges, we define its \textit{\textbf{traffic load}} to be representative of the number of packets in it:
\begin{equation}
\label{e:discrete_load}
\begin{array}{rcl} {\Lambda _{t} (X)} & {=} & {\sum _{e\in X}\ell (e)\tau (e) \quad } \\ 
{} & {=} & {\sum _{s,t\in V}\left(\sum _{e\in \left[s,t\right]\cap X}w(e) \right)\Lambda _{d} (s,t)\pi \left(\left[s,t\right]\right) } \\ 
{} & {=} & {\sum _{s,t\in V}\ell \left(\left[s,t\right]\cap X\right)\Lambda _{d} (s,t)\pi \left(\left[s,t\right]\right) } \end{array}
\end{equation} 

The above definition does not allow $X$ to be reduced to a vertex, but we can identify the smallest $\Lambda_t$-measurable neighbor of a vertex. 
Define the star of a vertex, ${\rm star}(x)$, to be the smallest subgraph of $G$ containing $x$;  
let $\ell(\mathrm{star}(x))$ be the sum of the lengths of 
set of edges abutting to $x$; the latter vertex set is denoted as $E_x=\left\{xy:y\in V,xy\in E\right\}$. Then consider

\noindent 

\[\frac{\Lambda _{t} \left({\rm star}(x)\right)}{\ell \left({\rm star}(x)\right)} =\frac{\sum _{e\in E_x}\Lambda _{t} (e) }{\sum _{e\in E_x}\ell (e) } \le \sum _{e\in E_x}\frac{\Lambda _{t} (e)}{\ell (e)}  =\sum _{e\in E_x}\tau (e)= :\tau (x)\] 
%
The inequality is a well known fact, and equality is achieved for the number of hops metric, that is, $w(e)=1,\forall e\in E$. Since the traffic in any edge connected to $x$ must be ``serviced'' by the ``hub'' $x$, the interpretation of $\tau (x)$ is the \textit{\textbf{traffic rate}} sustained by $x$. Therefore, the fundamental question addressed by this paper can be reformulated, from a very local point of view, as that of existence of vertices with very high traffic rate. The above string indicates that such vertices with very high traffic rates can be sought via the load measure $\Lambda_t$. 
The latter aspect will play a crucial role in Section~\ref{s:dif_geom}.

The traffic $\tau (x)$ can be computed by counting the number of geodesics $\left[s,t\right]$ having $x$ as a vertex. One should consider two classes of such geodesics, though: those that traverse $x$, that is, $x\ne s,\:x\ne t$, and those that either start at $x$, $\left[x,t\right]$, or terminate at $x$, $[s,x]$. The first ones have to be serviced twice, once in the input queue, once in the output queue; the others have to be serviced just once. There are $2(N-1)$ geodesics starting or terminating at $x$. The \textit{\textbf{betweenness centrality}} $\beta_c$ of a vertex $v$ is defined, for a uniquely geodesic graph, as the number of geodesics that have $v$ as a vertex (this includes those geodesics starting at $v$ or terminating at $v$). For uniformly defined demand, that is, $\Lambda _{d} (x,y)=1,\forall x\ne y$, the connection between the traffic rate and the betweenness is easily seen to be

\noindent 

\[\tau (v)=2(\beta_c (v)-2(N-1))+2(N-1)=2\beta_c (v)-2(N-1)\quad \]

\noindent 
\section{Basic conjectures}
\label{s:basic_conjectures}

For the sake of simplicity, 
we introduce a network curvature concept restricted to planar communication graphs 
and based on Alexandrov angles~\cite{BridsonHaefliger1999}. 

Let $(ab_{1} =ab_{\deg (a)+1} ,ab_{2} ,...,ab_{\deg (a)} )$ be a cyclic ordering of the set of edges attached to the vertex $a$. A geodesic \textit{triangle} is defined as $\Delta abc=[a,b]\cup [b,c]\cup [c,a]$. The \textit{Alexandrov angle} $\alpha _{k} $ at the vertex $a$ of the geodesic triangle $\Delta ab_{k} b_{k+1} $ is

\[\alpha _{k} ={\rm cos}^{-1} \frac{d(a,b_{k} )^{2} +d(a,b_{k+1} )^{2} -d(b_{k} ,b_{k+1} )^{2} }{2d(a,b_{k} )d(a,b_{k+1} )} . \] 
Then the (Gauss) \textit{curvature} at the vertex $a$  is defined as 

\[\kappa (a)=\frac{2\pi -\sum _{i=1}^{\deg (a)}\alpha _{k}  }{\sum _{k=1}^{\deg (a)}\area(\Delta ab_{k} b_{k+1} ) } \] 
where $\area(\Delta ab_{k} b_{k+1} )$ denotes the area of the geodesic triangle $\Delta ab_{k} b_{k+1}$, 
easily computable from the distances via Heron's formula. It is easily seen that, for the number of hops metric ($w(e)=1$), 
we have $\alpha _{k} =\pi /6$; therefore, $\kappa (a)<0$,  $\kappa (a)=0$, or $\kappa (a)>0$ depending on whether ${\rm deg}(a)>6$, ${\rm deg}(a)=6$, or ${\rm deg}(a)<6$, respectively. 

This simple definition is introduced to construct some easily understood illustrative examples. 
In the main body of the paper, though, 
we will use the Gromov definition of negative curvature for graphs   
(see Sec.~\ref{s:graphs_to_manifolds}.)

\begin{definition}[\cite{Jost1997}]
\label{d:inertia}
The \textit{\textbf{moment of inertia}} of a connected weighted graph $G$ 
with respect to a vertex $v$ is defined as 
$\phi^{(2)}_{G} (v)=\sum _{v^{i} \in V}d^{2} (v,v^{i} )$. 
\end{definition}

\noindent
Since the edges are treated as massless, this concept refers more to the underlying metric space structure $(V,d)$ 
than to the weighted graph structure $(V,E,d)$. 
Observe that this inertia may be infinite. 

\begin{definition}[\cite{Jost1997}]
\label{d:center_of_mass}
A \textit{\textbf{center of mass}} or \textit{\textbf{centroid}} of the weighted graph $G$ is defined as a vertex 
that achieves the infimum of the inertia: $\phi_G^{(2)}({\rm cm}(G))=\inf_{v\in V} \phi^{(2)}_{G} (v)<\infty$. 
\end{definition}
\noindent
If the graph is infinite ($|V|=\infty$) , this definition requires existence 
of a vertex $v$ such that $\phi_G^{(2)}(v)<\infty$. If we relax the minimum to be anywhere on the graph, 
it will in general be achieved on an edge, but restricting it to vertices makes it easier to relate it to the traffic. 

We can now pose our conjectures.

\begin{conjecture}[Negative curvature]
Consider a large but finite ($|V|<\infty$) negatively curved graph $G$, subject to uniformly distributed demand. Then 
\begin{enumerate} 
\item There are a very few nodes $v$ that have very high traffic rate $\tau(v)$ 
as measured by $\beta_c(v)$;  
furthermore, the vertices of highest traffic rate are in a small neighborhood of the vertices of minimum inertia. 
\item If the graph has a symmetry group that fixes some point $0$, 
this point achieves the unique minimum of the inertia and the maximum of the traffic rate. 
\end{enumerate}
\end{conjecture}
\noindent
The first part of this conjecture is illustrated on planar graphs in 
Section~\ref{s:examples}.  
As a model of a {\it large} negatively curved graph, 
it is tempting to take the full hyperbolic space $\mathbb{H}^n$. 
The latter is (globally) affine symmetric~\cite[Chap. IV]{Helgason},~\cite[Sec. 6.2]{Jost1998},~\cite[Chap. XI, Example 10.2]{KobayashiNomizu1996b}, 
and the transitivity of the symmetry group~\cite[XI, Th. 1.4]{KobayashiNomizu1996b}
would make the traffic uniformly distributed,  
if it weren't for the lack of convergence of the traffic function for uniformly distributed demand on infinite space. 
Restricting the network model to finite subsets of $\mathbb{H}^n$ breaks enough of the symmetry 
to create traffic spikes, even asymptotically as the size increases to infinity.  
An estimate of how close the center of mass and the point of maximum traffic is seems to be beyond 
our reach for general network models, at least at this stage. 
However, for hyperbolic balls, we {\it prove} that the two coincide 
and, as {\it major} result, that the sharpest traffic spike scales as $\vol(G)^2$. 

\begin{conjecture}[Nonnegative curvature]
Consider a large but finite ($|V|<\infty$) nonnegatively curved graph $G$ subject to uniformly distributed demand. Then
\begin{enumerate} 
\item Both the traffic and inertia functions $\tau$ and $\beta_c$ are more evenly distributed 
than in the case of a negatively curved graph.
\item If the graph has a vertex transitive symmetry group, then both the traffic and the 
inertia are uniform.  
\end{enumerate}
\end{conjecture}
\noindent
Even though this conjecture asserts that for both 
zero curvature graphs (e.g., Euclidean lattices) and positively curved graphs, 
the traffic is more smoothly distributed than in the case of negatively curved graphs, 
there is a significant difference between the two cases. 
A Euclidean lattice graph could be infinite, 
while positively curved (cubic) graphs are {\it finite}
by Higuchi's theorem~\cite{japanese_combinatorial_curvature,positively_curved_graphs,reti_4_combinatorial_curvature,mohar}.  
Positively curved graphs need not be truncated 
and hence enjoy more symmetry than truncated Euclidean or hyperbolic graphs. 
This point is easy to illustrate 
on the 1-skeleta of the boundaries of the Platonic solids, 
all of which are positively curved. 
(For those Platonic solids that have triangular faces, the above curvature formula applies; 
for the other Platonic solids, 
the more general Higuchi-Mohar-DeVos formula~\cite{mohar} should be applied.)   
The transitivity of the symmetry groups of the Platonic solids, 
together with the uniform distribution of the demand, 
implies the uniform distribution of the traffic on the boundary graphs. 
For more general positively curved graphs, 
this conjecture is proved using a Riemannian manifold model 
with its curvature bounded as $0< k_1^2 \leq \kappa(x) \leq k_2^2$. 
Here the symmetry is broken by the nonisotropic curvature

Truncating $\mathbb{E}^n$ to secure convergence results in symmetry breaking, 
creating traffic ``bumps,'' not as sharp as those of the negatively curved case, 
but sharper than those in the positively curved spaces.  
This statement is quantified by showing 
that the maximum traffic scales as $\vol(G)^{1+\frac{1}{n}}$.

\noindent 

\noindent 

\noindent 
\section{Some examples}
\label{s:examples}

We consider two simple examples: 
The first one (symmetric graph) highlights the worsening of the congestion as the curvature decreases. 
The second one (asymmetric graph) highlights the relationship between the curvature and the inertia. 

\subsection{Almost symmetric graph: congestion versus curvature}

We construct a graph from a single vertex (\#1)  
followed by the addition of seven neighbors (\#2-\#8) in a counterclockwise sense, 
as shown in Figure~\ref{f:congestion_versus_degree}, top. 
We then proceed from vertex \#2 and add vertices \#9, \#10, \#11, \#12, 
so that vertex \#2 has valence 7. We then add neighbors to vertex \#3, 
and proceed recursively, 
until we obtain the graph shown in Figure~\ref{f:congestion_versus_degree}, top. 
With 7 neighbors for each vertex, except the boundary ones, the graph is negatively curved, 
because $\sum_{i=1}^7 \alpha_i=7\times \frac{\pi}{3}>2 \pi$.  
The same construction can be done to generate graphs of valence 6 (vanishing curvature) 
and 8 (negative curvature). 
The plots of Figure~\ref{f:congestion_versus_degree}, bottom, clearly demonstrate 
that the maximum congestion worsens as the curvature becomes more negative. 

\begin{figure}[t]
\centering
\mbox{
\scalebox{0.5}{\rotatebox{-90}{\includegraphics{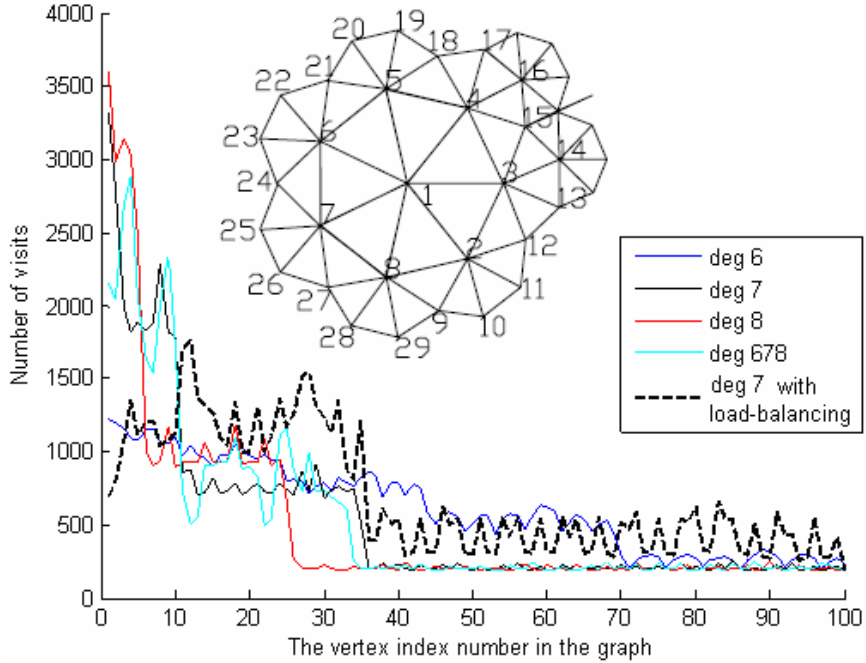}}}} 
\caption{Top: Almost symmetric negatively curved network of valence $7$. 
Bottom: Traffic function $\tau$ versus vertex number; 
clearly, congestion at centroid increases as curvature becomes more negative. 
Also shown is the traffic function after curvature-based load balancing 
(see~\cite{mingjithesis} for details).
}
\label{f:congestion_versus_degree}
\end{figure}

\subsection{Asymmetric graph: congestion versus inertia}

Here we consider a graph of valence 7, constructed the same way as in the preceding case, 
except that some ``appendices'' have been deliberately added to make the graph asymmetric. 
The results are shown in Figure~\ref{f:asymmetric_congestion_issues}. 
Because the graph is negatively curved, and as conjectured, the traffic has its maximum  
precisely at the point where the inertia is minimum. 
Observe that the traffic and the inertia are ``in opposite phase.'' 

\begin{figure}[t]
\centering
\mbox{
\scalebox{0.5}{\rotatebox{-90}{\includegraphics{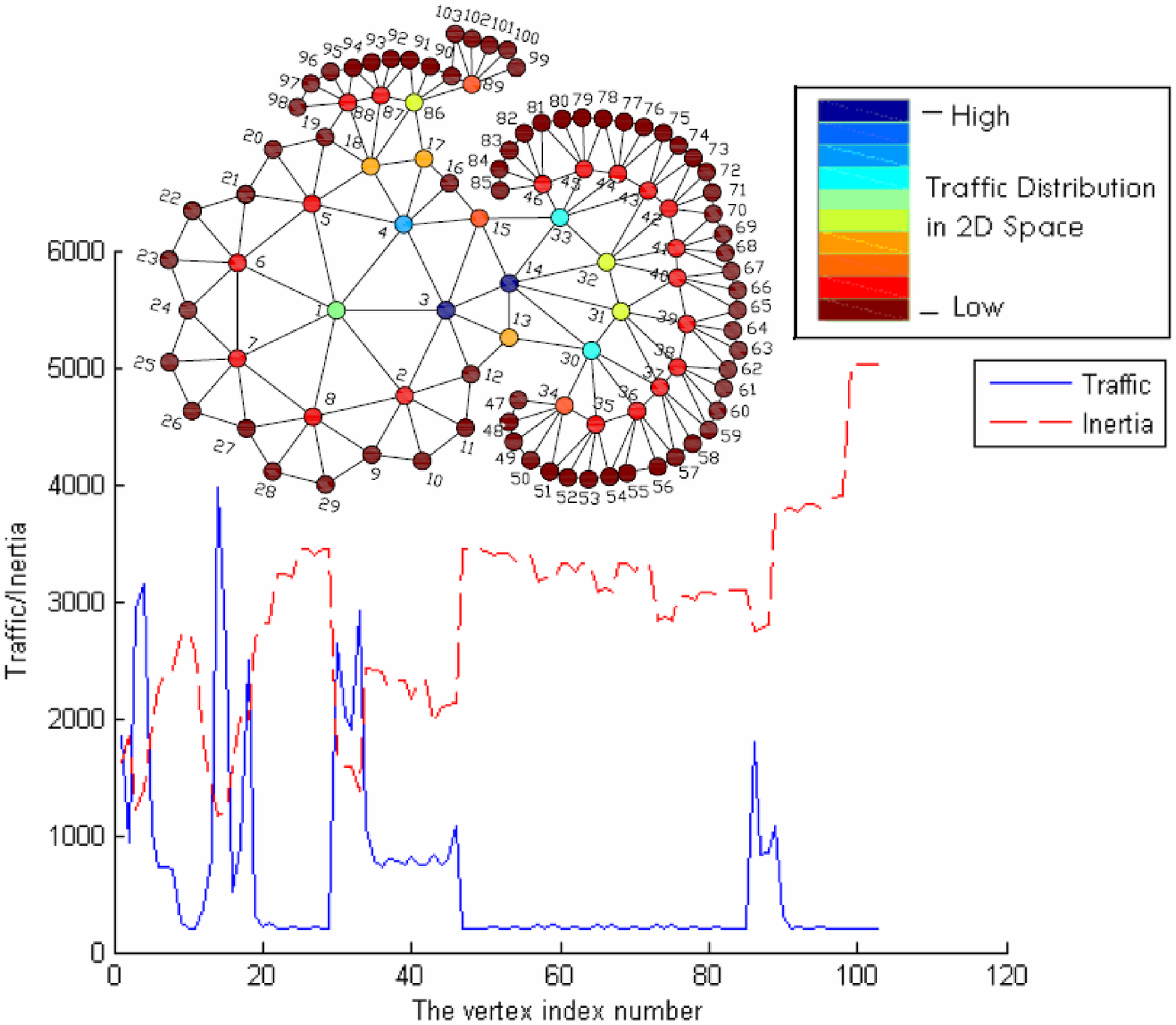}}} 
}
\caption{Traffic and inertia versus vertex number on a planar graph of valence $7$;    
the graph has been made asymmetric so that the conjecture could be verified 
on cases where the centroid does not coincide with the obvious symmetry center. 
The centroid occurs at vertex \#14. Also observe the heavy traffic at vertex \#4, 
because it is the gateway between the top appendix and the right side of the network.  
Observe the ``opposite phase'' phenomenon between the inertia and the traffic, 
clearly seen for vertex \#86.}
\label{f:asymmetric_congestion_issues}
\end{figure}

\section{From graphs to Riemannian manifolds}
\label{s:graphs_to_manifolds}

As said in the Introduction, it has become standard routine in complex networks to check their Gromov 
negatively curved property as a way to validate arguments based on such models as the Poincar\'e disk. 
For the sake of mathematical accuracy, we formulate the basic definition and result.

\begin{definition}
A geodesic metric space $(G,d_G)$, for example a graph, 
is said to be Gromov hyperbolic if there exists a $\delta < \infty$ 
such that every geodesic triangles $\triangle abc$ has an inscribed triangle $\triangle xyz$, 
$x \in [b,c], y \in [a,c], z\in [a, b]$, of a perimeter not exceeding $\delta$, that is, 
$d(x,y)+d(y,z)+d(z,x)\leq \delta$.  
\end{definition} 
This definition makes sense only for infinite graphs. 
Even though there exists a Gromov concept for finite graphs~\cite{scaled_gromov}, we really do not need the latter here, 
as our congestion analysis is asymptotic for very large graphs (or manifolds). 

\begin{theorem}[Bonk-Schramm~\cite{bonk_schramm}]
Let $(G,d_G)$ be a Gromov hyperbolic metric geodesic space with bounded growth at some scale. 
There exist an integer $n$, 
a convex subset $D \subseteq \mathbb{H}^n$, constants $\lambda,k$, and a map $f:G\rightarrow D$ 
such that 
$$ | \lambda d_G(u,v)-d_D(f(u),f(v))| \leq k, \quad \forall u,v \in G$$
and $\sup_{x\in D}d_D(x,f(G)) \leq k$. 
\end{theorem}
This theorem makes precise the somewhat loose statement made in the introduction regarding graph-manifold identification. 
Clearly, $G$, after a scaling $\lambda$, can be identified with $D$ via $f$ and subject to a bounded error $k$. 
The condition of bounded growth at some scale is satisfied, for example, when $G$ is a finite valence graph. 

Since this graph-manifold identification entails a bounded error, large scale problems on graphs can be mapped to more manageable continuous geometry problems on manifolds.

\noindent 

\noindent 
\section{Differential geometry proof of conjectures} 
\label{s:dif_geom}

To justify our numerical results related to congestion on planar graphs of uniform valence of 7, 8, 9, 
we could develop a Poincar\'e disk model (see Fig.~\ref{f:euclidean_poincare}, right panel). 
The latter is a faithful model 
in the sense that the graphs of valence 7, 8, 9 are quasi-isometric to the Poincar\'e disk. 
Recall that the Poincar\'e disk 
$\mathbb{D}=\left\{z=x+jy\in \mathbb{C}:\left|z\right|<1\right\}$ 
inherits its hyperbolic structure through the metric 
$ds^{2} =\frac{4dzd\bar{z}}{\left(1-\left|z\right|^{2} \right)^{2} } $. 
The latter leads to the area element $d\area=\frac{4dxdy}{\left(1-\left|z\right|^{2} \right)^{2} } $.

In the spirit of generalization, we develop models in the hyperbolic space $\mathbb{H}^n$, 
even in a Riemannian manifold $M$ with arbitrary Riemannian metric $ds^2$ and volume form $d\vol$,  
provided its sectional curvature bounded as $-k_2^2 \leq \kappa(x) \leq -k_1^2 <0$. 
To emphasize the role of the curvature, we also look at the Euclidean space $\mathbb{E}^2$ 
as a model of graphs of valence $6$ and further generalize the results to $\mathbb{E}^n$.  

\noindent 
\subsection{Hyperbolic versus Euclidean traffic load}
\label{s:euclidean_vs_hyperbolic}

Consider, in a Riemannian manifold $M$, a large ball $B_R(0)$ of radius $R$ with its center at the origin together with a convex subset $X\subset B_R(0)$. The \textit{\textbf{traffic load}} in $X$ is defined as (compare with~(\ref{e:discrete_load}))
\[\Lambda _{t} (X)=\iint \nolimits _{(s,t)\in B_R(0)\times B_R(0)}\ell \left(X\cap [s,t]\right)d\Lambda _{d} (s,t) \] 
As for the graph model, in this Riemannian context, 
the demand is {\it uniformly distributed} in the sense that $d\Lambda_d(s,t)$ is the 
product volume $d\vol(s)d\vol(t)$. 
The \textit{normalized} traffic load in $X$ is defined as follows:
\begin{equation}
\label{e:normalized_traffic_load}
\begin{array}{rcl} 
\lambda _{t} (X)&:=&\frac{\Lambda _{t} (X)}{{\rm vol}(B_R(0))^{2}}   \\
&=&\frac{1}{\vol\left(B_R(0)\right)^{2} } \iint \nolimits _{(s,t)\in B_R(0)\times B_R(0)}\ell \left(X\cap [s,t]\right)d\Lambda _{d} (s,t)  \end{array} 
\end{equation}

We take the hyperbolic plane $\mathbb{H}^{2} $ as a roughly similar model 
of an infinite negatively curved planar graph [9] and carry over the analysis to $n$ dimensions. 
In order to secure convergence of the traffic load, 
we restrict ourselves to the finite domain $B_R(0)\subset \mathbb{H}^n$,  
representative of a large but finite negatively curved graph $G$,  
and then we do the asymptotic analysis as $R \rightarrow \infty$. 

We now state our two major results. The proofs are in the Appendix: 
\begin{theorem}
\label{t:hyperbolic_traffic}
If $B_R(0)$ and $X=B_r(0)$, $r\ll R$, are concentric balls in $\mathbb{H}^n$ 
and if $d\Lambda _{d} (s,t)$ is the product volume measure, 
we have 
$$\lambda _{t} (B_r(0))\asymp c_1(n)r^n$$ 
where $c_k(n)>0$. Furthermore, 
in a Riemannian manifold of curvature bounded as $-k_2^2 \leq \kappa(x) \leq -k_1^2<0$, $\forall x \in M$, 
the asymptotic traffic load is bounded as $c_{k_1}(n) r^n \leq \lambda_t(X) \leq c_{k_2}(n)r^n$. 
\end{theorem}


\begin{theorem}
\label{t:euclidean_traffic}
If $B_R(0)$ and $X=B_r(0)$, $r\ll R$, are $n$-dimensional Euclidean balls (vanishing curvature) 
and if $d\Lambda _{d} (s,t)$ is the product volume measure, 
we have 
$$\lambda _{t} (B_r(0))\asymp c_0(n)\frac{r^n}{R^{n-1}}$$ 
for some constant $c_0(n)>0$, independent of $r,R$. 
\end{theorem}

Probably the most important conclusion to be drawn from the preceding two theorems is that 
the normalized traffic load in the small hyperbolic ball $B_r(0)$ remains bounded from below as $R\to \infty $,  
whereas the same normalized traffic but in the small Euclidean ball goes to zero as $R\to \infty $. 
Figure~\ref{f:euclidean_poincare} provides an intuitive explanation as to why this discrepancy happens. 
Because the hyperbolic geodesics are ``arched'' towards the center where the small ball lies, 
their average length in the small hyperbolic ball is much larger than in the small Euclidean ball. 

Figure~\ref{f:euclidean_poincare} gives a clue about elementary proofs of Theorems~\ref{t:euclidean_traffic}-\ref{t:hyperbolic_traffic} in $2$ dimensions.
Clearly, the natural parameterization of the $\s(x,y)$, $\t(x',y')$ points in the traffic load integral is via polar coordinates; 
the difficulty is to compute the Jacobian from $dxdydx'dy'$, $\frac{dxdydx'dy'}{(1-(x^2+y^2))^2(1-(x'^2+y'^2))^2}$, resp.,  
to the area squared element in polar coordinates of Euclidean, hyperbolic, resp., spaces.  
The details are available in~\cite{mingjithesis}; this elementary but more explicit proof yields 
a specific value for $c_0(2)=1/\pi$. The proof in the appendix is much more conceptual.

\begin{figure}[t]
\hskip-2cm\mbox{
\subfigure{\scalebox{0.3}{\rotatebox{-90}{\includegraphics{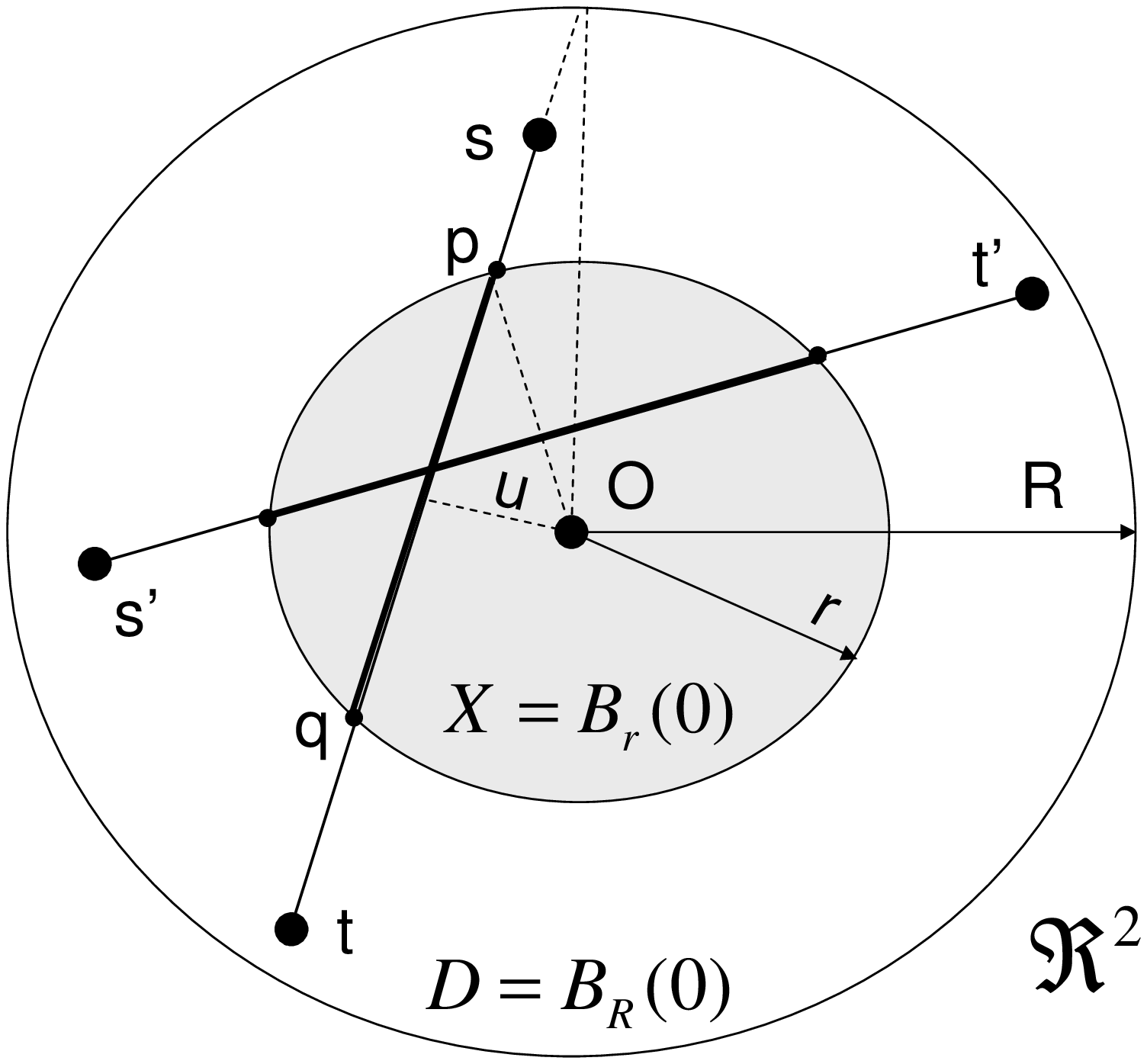}}}}
\subfigure{\scalebox{0.3}{\rotatebox{-90}{\includegraphics{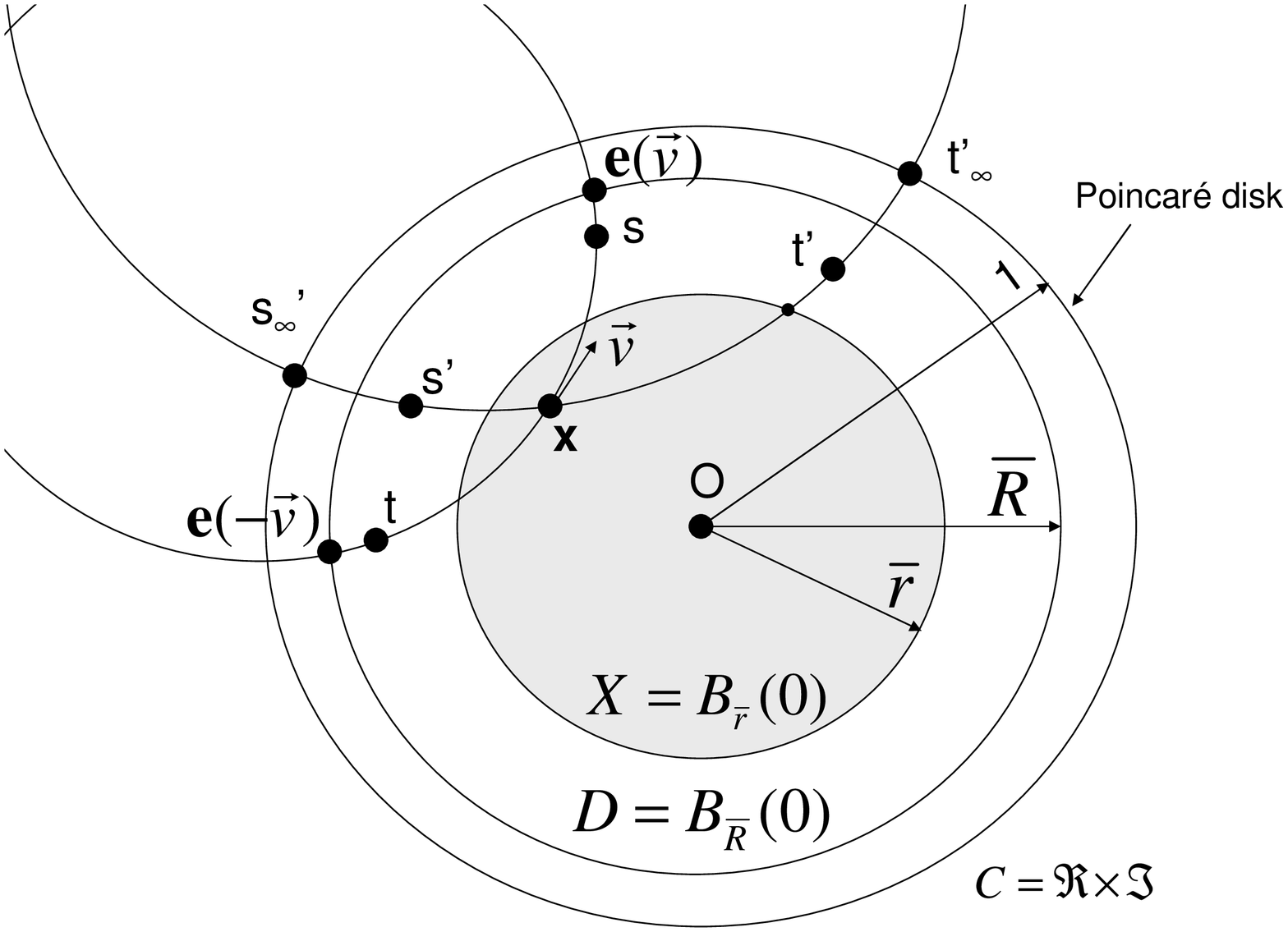}}}}
}
\caption{Traffic load in Euclidean space and Poincar\'e disk, 
the latter taken as a model of $\mathbb{H}^2$. The radii of the balls 
in the complex plane model are written $\bar{r},\bar{R}$ and the corresponding hyperbolic measurements are 
$r=\tanh^{-1}(\bar{r})$, $R=\tanh^{-1}(\bar{R})$.}
\label{f:euclidean_poincare}
\end{figure}

Observe that, besides its definition as the {\it traffic load}, 
$\lambda_t$ could be interpreted in another way. 
Instead of~(\ref{e:normalized_traffic_load}), consider the following:
$$ \frac{1}{\vol(B_R(0))^2} \iint_{B_R(0) \times B_R(0)} 
         I_{X \cap [s,t]}(s,t) d\Lambda_d(s,t) $$
where $I_{X \cap [s,t]}(s,t)=1$ if $X \cap [s,t] \not= \emptyset$ and $0$ otherwise. 
The above is clearly the proportion of communication paths transiting through $X$. 
If $X$ is between two balls, that is, $B_{r_1}(0) \subseteq X \subseteq B_{r_2}(0)$, 
the above is of the order of $\lambda_t(X)/r_i$. 
{\it Thus $\lambda_t(X)$ can be interpreted as the percentage of traffic passing through $X$.}

\noindent 

\noindent 
\subsection{Minimum inertia  }

Let $D \subset M$ be a convex domain of finite volume 
in a Riemannian manifold 
with distance $d(\cdot,\cdot)$.  
For $x \in D$, define (compare with Def.~\ref{d:inertia})
$$ \phi^{(p)}(x)=\int_D d(x,x')^p d\vol(x')$$
For $p=2$, the above is called \textbf{\textit{moment of inertia}} of $D$ relative to $x$. 
Next, assume there exists a point $\mathrm{cm}(D)\in D$ such that (compare with Def.~\ref{d:center_of_mass})
$$ \phi^{(p)}(\mathrm{cm}(D))=\inf_{x\in D} \phi^{(p)}(x) < \infty$$
For $p=2$, $\mathrm{cm}(D)$ is called a \textit{\textbf{center of mass}} or \textit{\textbf{centroid}} 
of $D$ (see~\cite[Def. 3.2.1]{Jost1997}). 
This concept was apparently introduced by Elie Cartan~\cite[p. 47]{Berger2000}. 

Arguments related to existence and uniqueness of the center of mass rely on strict convexity 
of $x \mapsto d(x,y)^p$ and $x \mapsto \phi^{(p)} (x)$. 
In negatively curved spaces, $p \geq 1$ suffices, 
whereas in nonpositively curved spaces (e.g., Euclidean spaces) the stricter condition $p>1$ is required. 
If $\phi^{(p)}$ is strictly convex, uniqueness is guaranteed~\cite[Lemma 3.1.1]{Jost1997}. 

\noindent 
\begin{theorem}
\label{t:min_inertia}
Let $M$ be a complete Riemannian manifold with its curvature bounded as 
$-k_2^2 \leq \kappa(x) \leq -k_1^2< 0$. 
Then the inertia of $B_R(0) \subset M$ relative to the point $x$,
\[\phi^{(2)} (x)=\int \nolimits _{B_{R}(0) }d(x,x')^2d\vol(x') \] 
has a unique minimum; furthermore, for $k_1=k_2$, this minimum is $x=0$. 
\end{theorem}

\begin{proof} 
Obviously, $M$ is a complete Riemannian manifold of nonpositive curvature, 
and hence it is a Busemann Non Positively Curved (NPC) space~\cite[page 45]{Jost1997}. 
Define the measure $\mu(\cdot)=\vol(\cdot) I_{B_R(0)}(\cdot)$, 
where $I_{B_R(0)}$ denotes the indicator of $B_R(0) \subset M$. 
Obviously, the measure $\mu$ has finite support and $\mu(B_R(0))<\infty$. 
Therefore, by~\cite[Th. 3.2.1]{Jost1997},
$$ \int_M d(x,x')^2 \mu(dx')=\int_{B_R(0)} d(x,x')^2 d\vol(x') $$
is finite and has a {\it unique} infimum. Furthermore, if $k_1=k_2=1$, 
$M$ has orientation preserving isometry group $SO^+(n,1)$. 
Under the subgroup that fixes $0$, $B_R(0)$ is invariant, 
and therefore $\phi^{(2)}$ has the same symmetries. 
Therefore, the only way to secure uniqueness of the infimum is $x^*=0$. 

\end{proof}

\noindent 

\noindent 
\subsection{Maximum load}

\begin{theorem}
\label{t:max_traffic}
Restricted to convex subsets $X$ of $B_R(0) \subset \mathbb{H}^n$ with the same hyperbolic area, 
$\lambda _{t} (X)$ reaches its maximum for a ball centered at the origin of $B_R(0)$. 
\end{theorem}

\noindent 

\begin{proof} 
The proof is a corollary of the proof of the Appendix. 
First, observe that $\mu_k(\x)$ is spherically invariant, that is, $\mu_k(\x)=\mu_k(S\x)$, 
$\forall S$ in the subgroup of $SO^+(n,1)$ that fixes $0$. 
Write $\mu_k(|\x|)$ to be the common value of $\mu(\x)$ when 
$\x \in \partial B_{|x|}(0)$. 

Next, we want to prove that 
\[\int_{S_\x^{n-1}} \int_0^{e(\vec{v})}\int_0^{e(-\vec{v})}\sh_k^{n-1}(x+y)dxdyd\vec{v} \]
subject to $e(\vec{v})+e(-\vec{v})=E$ is maximum for $e(\vec{v})=e(\vec{-v})$. 
This is proved using the augmented functional 
\[\int_{S_\x^{n-1}} \int_0^{e(\vec{v})}\int_0^{e(-\vec{v})}\sh_k^{n-1}(x+y)dxdy 
 +L(e(\vec{v})+e(-\vec{v}))\]
where $L$ is the Lagrange multiplier. 
Setting the partial derivatives relative to $e(\vec{v})$ and $e(-\vec{v})$ to zero yields
\begin{eqnarray*}
\int_{S_\x^{n-1}} \int_0^{e(-\vec{v})}\sh_k^{n-1}(e(\vec{v})+y)dy+L&=&0\\
\int_{S_\x^{n-1}} \int_0^{e(\vec{v})}\sh_k^{n-1}(x+e(-\vec{v}))dx+L&=&0
\end{eqnarray*}
Clearly, the above yields $e(\vec{v})=e(-\vec{v})$ 
and the maximum is monotone increasing with $E$. It follows that 
\[\mu_k(|\x|)\leq \frac{1}{\vol(B_R(0))^2}
\int_{S_\x^{n-1}} \int_0^R\int_0^R\sh_k^{n-1}(x+y)dxdyd\vec{v}\]
But 
\[\mu_k(\textbf{0})=\frac{1}{\vol(B_R(0))^2}
\int_{S_0^{n-1}} \int_0^R\int_0^R\sh_k^{n-1}(x+y)dxdyd\vec{v}\]
Thus $\mu_k(|\x|)$ reaches its maximum at $\x=0$. 
Furthermore, it is obviously symmetric for the subgroup of $SO^+(n,1)$ that fixes $0$. 
Finally, from the Appendix, $\mu_k(||\x||)$ is monotone decreasing with $||\x||$. 

It remains to show that the optimal way to distribute the volume allocated to $X$ 
is in a ball around $0$. First, we show that, if $\inf \lambda_t(X)$ is achieved for some $X^*$ 
that does not enjoy the symmetry under the action of the subgroup of $SO^+(n,1)$ that fixes $0$, 
then the same $\lambda_t(X^*)$ can be achieved for another subset, that has the symmetry, but that cannot be optimal.  

Decompose the big ball as $B_R(0)=\sqcup_{i=0}^{m-1} A[r_i,r_{i+1})$, 
where $A[r_i,r_{i+1})$ is the annulus $\{x: r_i \leq d(0,x) < r_{i+1}\}$. 
Clearly, there exists an annulus $A[\rho_i,\rho_{i^+}) \subseteq A[r_i,r_{i+1})$ 
such that $\lambda_t(X \cap A[r_i,r_{i+1}))=\lambda_t(A[\rho_i,\rho_{i^+}))$. 
Thus $\lambda_t(X)=\sum_i \lambda_t(A[\rho_i,\rho_{i^+}))$ and  
$$\lambda_t(X)=\sum_i \int_{\rho_i}^{\rho_{i^+}} \mu_k(r) T(r)dr $$
where $T(r)$ is the transverse measure such that $\vol(A[r,r+dr))=T(r)dr$, 
that is, $T(r)$ is the ``area'' of the sphere at $0$ with radius $r$.  
Consider now the constrained optimization problem
$$ \inf_{\rho} \sum_i \int_{\rho_i}^{\rho_{i^+}} \mu_k(r) T(r)dr $$
subject to
$$ \sum_i \int_{\rho_i}^{\rho_{i^+}} T(r)dr =\vol(X) $$
Again, a Lagrange multiplier argument proves that optimality could not hold 
with some $i$ such that $\rho_{i^+} < \rho_{i+1}$. (The intuition is that,  
since the density $\mu_k$ is monotone decreasing, optimality would require $\rho_{i+1}$  
to drop down to $\rho_{i^+}$.) Thus $\rho_0=0$ and $\rho_{i^+}=\rho_{i+1}$, that is, the presumed optimal annulus 
collapses to a ball.
\end{proof}

\noindent 

\noindent 
\section{Differential geometry versus real network congestion}
\label{s:real_networks}

It is argued that, no matter how theoretical our model $\Lambda _{t} (B_r(0))$ of the traffic load is, 
it is remarkably accurate at predicting how the ``load at the center,'' here $0$, scales with the number of vertices, $N$, 
in a \textit{real} network. 

In~\cite{arXiv_dmitri}, 
the traffic load at the ``center'' of a sample of networks from the Rocketfuel data base~\cite{rocketfuel} 
has been numerically found to scale as $N^{2}$. 
Using the scaled Gromov analysis~\cite{scaled_gromov}, 
it was asserted in~\cite{arXiv_dmitri} that the networks from that data base are negatively curved. 
In~\cite{arXiv_dmitri}, the ``center'' is somewhat loosely defined as \textit{the ``network core,''} 
or the \textit{``set of nodes that are at the intersection of the majority of geodesics.''}  
The latter intuitive concept is unmistakably the same as our theoretical \textit{center of mass} concept. 
In the same paper~\cite{arXiv_dmitri}, Narayan and Saniee provide experimental evidence that the traffic load at the center of Watts-Strogatz Small-World $2$-dimensional networks scales as $N^{1.5}$. It is commonly admitted, and it has been proved using the scaled Gromov $\delta$-analysis~\cite{scaled_gromov}, that those networks are not hyperbolic; they are rather Euclidean, even positively curved, in a certain range of the connectivity and rewiring parameters~\cite[Sec. 6.4.2]{Matt_thesis}.

\begin{fact}[Narayan and Saniee~\cite{arXiv_dmitri}]
\label{fact:experimental}
Let $\beta_c(v)$ be the betweenness centrality of the vertex $v$ in a network. 
\begin{enumerate}
\item In the Rocketfuel data base~\cite{rocketfuel} of real networks, 
which are scaled-Gromov hyperbolic by the definition of~\cite{scaled_gromov}, 
the maximum traffic rate scales as 
$$ \max_v \beta_c(v) = \Theta(N^2) $$
\item For synthetic 2-dimensional Euclidean lattice networks,  
the  maximum traffic rate scales as
$$ \max_v \beta_c(v) = \Theta(N^{1.5}) $$
\end{enumerate}
\end{fact}

To proceed to a continuous geometry justification of the above, 
the graphs are embedded in appropriate manifolds, 
using the Bonk-Schramm theorem for a Gromov hyperbolic graph, 
or using the trivial Euclidean embedding for the Euclidean lattice graphs. 
The Bonk-Schramm embedding $f : G \rightarrow D \subseteq \mathbb{H}^n$ maps 
the geodesic flow on the graph $G$  
to a quasi-geodesic flow on the manifold $D$. 
Consider a vertex $v$ on the graph. We clearly have $\Lambda_t(\mbox{star}(v)) \approx \Lambda_t(B_r(f(v)))$, 
for some $r$ of the order of the mesh of the graph.  
But as argued in Section~\ref{s:basic_def}, $\Lambda_t(\mbox{star}(v))$ scales as $\tau(v)$, 
which scales as $\beta_c(v)$. 
The continuous geometry model 
of the traffic metric of the above {\bf Fact} is therefore $\Lambda_t(B_r(x))$ for some $x \in D$. 
Obviously, the continuous geometry equivalent of $N$ is 
$$ N=\vol (D) $$

\begin{fact}[Theorems~\ref{t:hyperbolic_traffic} and \ref{t:euclidean_traffic}]
\label{fact:theoretical}
Let $\Lambda_t(B_r(x))$ be the traffic load in a small metric ball $B_r(x)$ embedded in a large metric ball $B_R(0)$, 
where $R \gg r$. 
\begin{itemize}
\item In a hyperbolic ball $B_R(0) \subset \mathbb{H}^n$, we have 
$$ \max_x \Lambda_t(B_r(x))=\Theta(\vol(B_R(0))^2)=\Theta(N^2) $$
\item In a Euclidean ball $B_R(0) \subset \mathbb{E}^n$, we have 
$$\max_x \Lambda_t(B_r(x))=\Theta \left( \frac{\vol(B_R(0))^2}{R^{n-1}}\right)=\Theta\left(N^{1+\frac{1}{n}}\right)$$
\end{itemize}
\end{fact}

\textit{In view of the striking consistency between Fact~\ref{fact:experimental} and Fact~\ref{fact:theoretical}, 
our model correctly predicts how the maximum traffic load scales as a function of $N$ 
or the volume of the manifold.} 

By the same token, we have extended the Euclidean results of~\cite{arXiv_dmitri} to the case where $n>2$. 
This generalization allows us to observe that,  
as the dimension gets higher and higher, 
the Euclidean congestion decreases and the gap between traffic loads in hyperbolic and Euclidean spaces increases. 

While there are computational methods 
to check the Gromov property of a network~\cite{scaled_gromov,Matt_thesis,eurasip_clustering,4_point}, 
associating a dimension with a complex network is an entirely other matter. 
The remarkable feature is that the asymptotic traffic analysis in negatively curved spaces 
\textit{transcends the dimension}. 

\section{The case of positive curvature}

The case of a Riemannian manifold positively curved as $0 < k_1^2 < \kappa < k_2^2$ is significantly different from that of a nonpositively curved manifold, 
because the diameter of the former is bounded as $\pi/k_1$. Furthermore, by the sphere theorem~\cite[Chap. 13]{docarmo}, 
if $k_1=k_2/2$, the manifold is homeomorphic to a sphere. 
Thus, contrary to the nonpositive curvature case, we cannot take a very large manifold with fixed positive curvature;  
so, we will have to take an arbitrarily large manifold of positive curvature decreasing to $0$. 
The first part of the following theorem is completely trivial; the second part is easily proved 
by bounding the integrand of~(\ref{e:mu_k_plus}) and taking the upper limit of the two inner integrals to be 
$\pi R=\pi/k$. 

\begin{theorem}
Let $M$ be a connected $n$-dimensional Riemannian manifold with constant curvature $\kappa=k^2>0$. 
\begin{enumerate}
\item Both the normalized traffic $\lambda_t$ and the inertia $\phi^{(2)}$ over the whole manifold are uniform.
\item $\lambda_t(B_r(0))=\Theta\left(\frac{1}{R^{2n-2}}\right)$, 
where $B_{\pi R}(0)=M\setminus \{\mbox{antipodal point of } 0\}$.  
\end{enumerate} 
\end{theorem}
\noindent
Comparing with Theorem~\ref{t:euclidean_traffic}, 
it is clear that the normalized traffic decays with $n$ even faster than in the Euclidean case. 

Finally, using the comparison formula of the Appendix, the following is easily derived:
\begin{theorem}
Let $M$ be a connected Riemannian manifold with curvature bounded as $k_1^2 < \kappa \leq k_2^2$. 
Then
$$ \mu^+_{k_2}(\x) < \frac{d\lambda_t(\x)}{d\x} \leq \mu^+_{k_1}(\x) $$
\end{theorem}
\noindent
Thus, under varying but bounded curvature, 
the traffic density in a small neighborhood of $\x$ remains bounded 
between the densities in fixed curvature. 
(Observe that by symmetry $\mu^+_{k_i}(\x)$, $i=1,2$,  are independent of $\x$.) 

\section{Conclusion}

We have provided a mathematical justification of the experimentally observed fact that 
negatively curved networks---even with uniform curvature---driven by uniformly distributed demand 
have small areas of very high traffic concentration. 
Nonnegatively curved networks, on the other hand, do not exhibit this phenomenon 
as dramatically as negatively curved networks.  
In fact, uniformly positively curved networks have uniform traffic distribution; 
more generally, bounds on the curvature implies bounds on the maximum traffic. 

Since the root cause of congestion in a network is its negative curvature, load balancing could be achieved 
by {\it controlling the curvature} to become and remain positive~\cite{mingjithesis}, despite outages 
and varying demand. 

The areas of maximum traffic have been narrowed down to areas of low inertia.  
It has been {\it proved} that networks with enough symmetry  
have colocated maximum traffic and minimum inertia. 
But as already said, for general networks, bounding the distance between the two points is a challenging problem. 

The traffic dealt with here is the one driven by a uniformly distributed demand. 
The extension to nonuniformly distributed demand is not a major hurdle. 
It suffices to redefine the inertia as $\phi^{(2)}_G(v)=\sum_{v^i}d^2(v,v^i)\left(\sum_j \Lambda_d(v^i,v^j)\right)$. 

Finally, the present analysis is a spatial one. 
The temporal component would bring the dynamics of packet drops and retransmission into the picture. 
Early {\tt ns-2} simulation have shown that UDP traffic has its  
maximum packet drop at the point of maximum traffic/minimum inertia~\cite{mingjithesis}.

\noindent 

\noindent 
\section{Appendix: Nonpositively curved spaces}
\label{s:app}

Let $D$ be a convex domain in 
an $n$--dimensional Riemannian manifold $M$. Here convex means that, for every $\s$ and $\t$ in $D$, 
there is a unique shortest geodesic, or shortest path, $[\s,\t]$ going from $\s$ to $\t$ and contained in $D$. 

For a subset $X$ of $D$, we are interested in
$$
\lambda_t(X) =
\frac1{\mathrm{vol}(D)^2} 
\int_{D\times D} \ell(X\cap [\s,\t]) \, d\s \, d\t
$$
where the integral is with respect to the square power of the $n$--dimensional volume of $D$. 

From a dynamical point of view, assuming uniform traffic between pairs of points of $D$, 
and assuming that this traffic travels at unit speed along geodesics, 
$\lambda_t(X)$ measures the average of the amount of time spent in $X$. 

We provide an estimate for $\lambda_t(X)$ when $M$ has non-positive curvature, 
and an exact computation when the curvature is constant. 

Consider a point $\s$ in $D$.
For every unit vector $\vec{v}$ based at $\s$,  
draw the geodesic $g_{\vec{v}}$ emanating from $\s$ in the direction of $\vec{v}$, 
and let $e(\vec{v})$ be the distance from $\s$ to the point  where the geodesic exits $D$. 
(This exit point is unique by convexity of $D$.) 
We can similarly consider the geodesic $g_{-\vec{v}}$ emanating from $\s$ in the opposite direction $-\vec{v}$, 
and the corresponding distance $e(-\vec{v})$.

For $k>0$, define 
\begin{equation}
\label{e:mu_k}
\mu_k(\s) = \frac 1 {\mathrm{vol}(D)^2} 
\int_{S^{n-1}_\s} 
\int_0^{e(\vec{v})}
\int_0^{e(-\vec{v})}
{\textstyle \frac1{k^{n-1}}}
\sinh^{n-1}(kx+ky)
\,  dx\, dy \, d\vec{v}
\end{equation}
where the outer integral is with respect to the $(n-1)$--dimensional volume of the unit sphere $S^{n-1}_\s$ 
consisting of all unit vectors $\vec{v}$ based at $\s$. 
The inner double  integral can of course be computed by elementary calculus. 

For $k=0$, let
$$
\mu_0(\s) = \frac 1 {\mathrm{vol}(D)^2} 
\int_{S^{n-1}_\s} 
\int_0^{e(\vec{v})}
\int_0^{e(-\vec{v})}
(x+y)^{n-1}
\, dx\, dy \, d\vec{v}
$$

\begin{theorem}
If the curvature of $M$ is constant and equal to $-k^2$, then 
$$\lambda_t(X) = \int_X \mu_{k} (\x) \, d\x $$
More generally,  if the curvature of $M$ is everywhere
bounded between two non-positive constants $-k_1^2$ and $-k_2^2$ with $0\leq k_1\leq k_2$,
$$
\int_X \mu_{k_1} (\x) \, d\x 
\leq \lambda_t(X) \leq
\int_X \mu_{k_2} (\x) \, d\x .
$$
\end{theorem}
In other words, the quantity $\lambda_t(X)$ is estimated by the integral of the density functions $\mu_k$.

\begin{proof}
Let $\mathcal T$ be the set of triples $(\s,\t,\x)$ where $\s$, $\t$, $\x$ are points of $D$ 
such that $\x$ is located on the geodesic $[\s,\t]$ and inside of $X$. 
There are two natural measures that can be put on $\mathcal T$.

The first one is the product $d\s\,d\t\,d\ell$ of the volume in $\s$, the volume in $\t$, 
and the arc length parameter $\ell$ along the geodesic $[\s,\t]$. 
The reason why we are considering $\mathcal T$ with this measure is that
$$
\int_{D\times D} \ell(X\cap [\s,\t]) \, d\s \, d\t
= \int_{\mathcal T} d\s\,d\t\,d\ell .
$$

The second one requires a different description of $\mathcal T$. 
If we consider the unit vector $\vec{v}$ based at $\x$ and pointing in the direction of $\s$ 
(so that $-\vec{v}$ points in the direction of $\t$), the distance $s$ from $\x$ to $\s$, 
and the distance $t$ from $\x$ to $\t$, 
then the three points $(\s,\t,\x)$ are completely determined by the $\x$-based vector $\vec{v}$ 
and by the numbers $s$ and $t \geq0$. 

Recall that the \emph{unit tangent bundle} $T^1X$ consists of all unit vectors $\vec{v}$ based at points of $X$. 
This is a $(2n-1)$--dimensional manifold ($n$ dimensions for the base point, $n-1$ for the direction), 
and the  metric of $X$  naturally lifts to a Riemannian metric on $T^1X$ by using the Levi-Civita connection. 

The above construction identifies the set $\mathcal T$ of triples $(\s,\t,\x)$ with $\x\in X\cap [\s,\t]$ 
to the subset of $T^1X \times \mathbb{R} \times \mathbb{R}$ consisting of those $({\bf v}, s, t)$ 
where $\v$ is a point of $T^1X$, and where $0\leq s \leq e(\vec{v})$ and $0\leq t \leq e(-\vec{v})$, 
with $e(\vec{v})$ the distance between $\x$ and the point $\e(\vec{v})$ where the geodesic starting at $\x$ along the direction $\vec{v}$ exits $D$ 
(see Fig.~\ref{f:euclidean_poincare}, right panel).  
Considering the volume form $d{\v}$ on $T^1X$ now provides another measure $d{\v}\,ds\,dt$ on $\mathcal T$. 

When the curvature of $M$ is bounded between $-k_1^2$ and $-k_2^2$, 
the two measures $d\s\,d\t\,d\ell$ and $d{\v}\,ds\,dt$ on $\mathcal T$ can be compared 
by the standard Riemannian arguments on the variation of geodesics (using Jacobi fields):
\begin{lemma}
$$
\sh_{k_1}^{n-1}(s+t) \, d{\v}\,ds\,dt 
\leq  d\s\,d\t\,d\ell 
\leq \sh_{k_2}^{n-1}(s+t) \, d{\bf v}\,ds\,dt
$$
where $\sh_k(x) = \frac1k \sinh(kx)$ if $k >0$ and $s_0(x) = x$. 
\end{lemma}

Therefore,
\begin{align*}
\lambda_t(X) &= \frac 1{\mathrm{vol}(D)^2}  \int_{\mathcal T} d\s\,d\t\,d\ell \\
&\leq  \frac 1{\mathrm{vol}(D)^2}  \int_{\mathcal T} \sh_{k_2}^{n-1}(s+t) \,ds\,dt\, d{\v}\\
&=  \frac 1{\mathrm{vol}(D)^2}  \int_{T^1X} 
\int_0^{e(\vec{v})}
\int_0^{e(-\vec{v})}
\sh_{k_2}^{n-1} (s+t)
\,  ds\, dt\, d{\v}
    \\
 &=  \frac 1{\mathrm{vol}(D)^2}  \int_{X} \int_{S^{n-1}_\x}
\int_0^{e(\vec{v})}
\int_0^{e(-\vec{v})}
\sh_{k_2}^{n-1} (s+t)
\, ds\, dt \, d\vec{v}\, d\x
  \\
& = \int_X \mu_{k_2} (\x) \,d\x
\end{align*}
Observe that, on the fourth line, 
$d\vec{v}$ stands for the $(n-1)$--dimensional volume on the sphere $S^{n-1}_\x$ of unit vectors based at $\x$, 
whereas  $d{\v}$ represents the $(2n-1)$--dimensional volume of $T^1X$ in the two lines before. 

The inequality 
$$
\lambda_t(X)  \geq  \int_X \mu_{k_1} (\x) \,d\x
$$
is proved by the same argument. 
\end{proof}

For instance, consider the case where $D$ and $X$ are two concentric balls in a $n$-manifold of constant curvature $-k^2\leq 0$, of respective radii $R$ and $r$ with $r<\kern -3pt<R$. 
Then, $e(\vec{v})\asymp R$ for every vector $\vec{v}$ based at a point of $X$.

\begin{theorem}
If $D$ and $X$ are two concentric balls of respective radii $R$ and $r$, with  $r<\kern -3pt<R$, in the Euclidean space of dimension $n$, then the proportion $\lambda_t(X)$ of traffic in $D$ that transits through $X$ is approximately
$$
\lambda_t(X) \asymp c_0(n) \frac{r^n}{R^{n-1}}.
$$
where $c_0(n)$ is an explicit constant depending on the dimension $n$. 

If, instead, $D$ and $X$ are concentric balls in the $n$--dimensional  hyperbolic space (where the curvature is $-1$), then
$$
\lambda_t(X) \asymp c_1(n)  r^n.
$$

\end{theorem}

\begin{proof}
Under these hypotheses, 
$\lambda_t(X)$ is of the order of $\mu_k(\textbf{0})\vol(X)$.
Write $\mu_k(\textbf{0})=M_k(\textbf{0})/\vol(D)^2$, where $M_k(\textbf{0})$ is the triple integral. 
$e(\vec{v})\asymp R$ for every vector $\vec{v}$ based at a point of $X$. 
Therefore, in the Euclidean case, $M_0(\textbf{0})$ is of the order of $R^{n+1}$,  
$\mathrm{vol}(D)$ is of the order of $R^n$, and $\lambda_t(X)$ is of the order of $r^n/R^{n-1}$. 
In the hyperbolic case, $M_1(\textbf{0})$ is of the order of $e^{2(n-1)R}$, 
$\vol(D)$ is of the order of $e^{(n-1)R}$, 
and thus $\lambda_t(X)$ is of the order of $\vol(X)$,  
which is itself of the order of $r^{n}$ for $r$ small. 
\end{proof}

Note the independence on $R$ in the negatively curved case.  

\noindent\textbf{Proof of Lemma:}
The Jacobi field equation~\cite{docarmo} in the hyperplane orthogonal to the geodesic $[\x,\e(\vec{v})]$ 
in the $n$-manifold $M$ of constant curvature $-k^2$ reads $\frac{d^2J(s)}{ds^2}-k^2J(s)=0$. 
If $dx_i=J(0)$, $i=1,...,n-1$, is a linear element orthogonal to the geodesic at $\x$ 
and $d\theta_i=J'(0)$ is the elementary angle on the sphere $S_\x^{n-1}$, 
the solution to the Jacobi field equation at $s$ reads, with a similar solution for $[x,\e(-\vec{v})]$ at $t$, 
\begin{eqnarray*}
d\s_i &=& ~~\frac{1}{k} \sinh (ks) d \theta_i + \cosh (ks) dx_i   \\
d\t_i &=&  -\frac{1}{k} \sinh (kt) d \theta_i + \cosh (kt) dx_i  
\end{eqnarray*}
In the above $d\s_i$, $d\t_i$ are the variations of the geodesic in the orthogonal hyperplane at $\s$, $\t$, respectively.
The volume defined by the variation of $\s$ at $s$ is  $d\s=\left(\wedge_{i=1}^{n-1}d\s_i \right) ds$.   
Clearly, $\wedge_{i=1}^{n-1}d\s_i$ is the transverse variation and $ds$ is the longitudinal variation.  
A similar statement holds for $d\t$.
Because of the skew-symmetry of exterior differential forms, it is convenient to introduce the notation 
$\theta \stackrel{\sigma}{=}\omega$ to denote $\theta =(-1)^{n-1}\omega$. 
With this notation, we get
\begin{eqnarray*}
d\s\wedge d\t \wedge d\ell &\stackrel{\sigma}{=}& \wedge_{i=1}^{n-1}\left( d\s_i \wedge d\t_i \right) \wedge ds \wedge dt \wedge d\ell \\
 &\stackrel{\sigma}{=}& \wedge_{i=1}^{n-1} \left(\frac{1}{k}\sinh (k(s+t)) d\theta_i \wedge d\x_i \right) \wedge ds \wedge dt \wedge d\ell \\
 &\stackrel{\sigma}{=}&\left(\frac{1}{k}\sinh (k(s+t))\right)^{n-1}\left(\wedge_{i=1}^{n-1} d\theta_i \wedge d\x_i \right)\wedge d\ell \wedge ds \wedge dt\\
&=& \sinh_k^{n-1}(k(s+t))(\wedge_{i=1}^{n-1} d\theta_i)\wedge \left(\left(\wedge_{i=1}^{n-1} d\x_i\right) \wedge d\ell\right) \wedge ds \wedge dt\\
&=& \sinh_k^{n-1}(k(s+t))\left(d\vec{v} \wedge ((\wedge_{i=1}^{n-1} d\x_i) \wedge d\ell) \right)\wedge ds \wedge dt\\
&=& \sinh_k^{n-1}(k(s+t))\left(d\vec{v} \wedge d\x\right) \wedge ds \wedge dt\\
&=& \sinh_k^{n-1}(k(s+t))d\v \wedge ds \wedge dt
\end{eqnarray*}
Observe that, to go from the second to the third line, we need skew-commutativity of $\wedge$ along with some elementary hyperbolic trigonometry. 
The inequality follows from the Rauch comparison argument~\cite[VIII, Th. 4.1]{KobayashiNomizu1996b}, 
which in this context refers to the monotone increasing property of the solution $\sh_k(x)$ 
to the Jacobi field equation with $k$.
$\blacksquare$

We finish this Appendix with two corollaries, useful in the main body of the text. 
A notation needs to be made more explicit: $e_\x(\vec{v})$ denotes the length of the geodesic 
shot from $\x$ in the direction $\vec{v}$ and terminating at $\partial B_R(0)$. 
\begin{corollary}
For fixed $\vec{v}$, 
$e_\x(\vec{v})+e_\x(-\vec{v})$ is strictly monotone decreasing with $||\x||$. 
\end{corollary}
\begin{proof}
Let $\e_\x(\vec{v})$ be the point where the geodesic starting at $\x$ with direction $\vec{v}$ exits $B_R(0)$. 
With a mild abuse of notation, we will let $\vec{v}$ denote the angle between the radial $[\textbf{0},\x]$  
and the tangent to the geodesic at $\x$. Same definition applies to $\e_\x(-\vec{v})$. 
Because $\mathbb{H}^n$ is isotropic, 
the geodesic $[\e_\x(-\vec{v}),\e_\x(\vec{v})]$ is contained in a plane, 
which itself contains $\x$ and $\textbf{0}$. 
Let [\textbf{0},\textbf{p}] be the radial orthogonal to that geodesic. 
In the right angle triangle, $\triangle \textbf{0}\textbf{p}\e_\x(\vec{v}))$, 
we get 
$\cosh \left(\frac{1}{2}||\e_\x(\vec{v})\e_\x(-\vec{v}))||\right)=\frac{\cosh R}{\cosh (||\textbf{0}\textbf{p}||)}$. 
On the other hand, in the right angle triangle $\triangle \x\textbf{p}\textbf{0}$, 
we have $\sinh ||\textbf{0}\textbf{p}||=\cosh(x) \sin \vec{v}$. 
It follows that 
$$ \cosh\left(\frac{1}{2}||\e_\x(\vec{v})\e_\x(-\vec{v}))||\right)=\frac{\cosh R}{\sqrt{1+\cosh^2x \sin^2 \vec{v}}}$$
and from there the result is obvious. 
\end{proof}

\begin{corollary}
For $k>0$, the density function $\mu_k(\s)$ is strictly monotone decreasing with $||\s||$.
\end{corollary}

\begin{proof}
It is easily seen that the double inner integral in the definition~(\ref{e:mu_k}) of $\mu_k$  
is a polynomial in the variables $\exp(\pm ki(e_\s(\vec{v})+e_\s(-\vec{v}))$, where $i$ is some power. 
From the same definition, it is also obvious that the inner double integral is monotone increasing with 
$e_\s(\vec{v})+e_\s(-\vec{v})$. But, by the previous corollary, the latter is monotone decreasing with $||\s||$. 
Thus the double inner integral in~(\ref{e:mu_k}) is monotone decreasing with $||\s||$, and so is the average 
over the sphere $S_\s^{n-1}$. 
\end{proof}

\section{Appendix: Positively curved spaces}
\label{s:app_positive}

The case of uniformly positive curvature $\kappa=k^2>0$ is treated in a way parallel to the preceding one. 
The difference resides in the Jacobi field equation, 
which now takes the form $\frac{d^2 J(s)}{ds^2}+k^2J(s)=0$,  
with ordinary trigonometric function, rather than hyperbolic trigonometric function, solution. 
Therefore, if for some bounded subset $D$ of the space of curvature $k^2$ we define 
\begin{equation}
\label{e:mu_k_plus}
\mu^+_k(\s)=\frac{1}{\vol(D)^2}\int_{S^{n-1}_\s}\int_0^{e(\vec{v})}\int_0^{e(-\vec{v})} \frac{1}{k^{n-1}}\sin^{n-1}(k(x+y))dxdyd\vec{v} 
\end{equation}
the normalized traffic load in $X \subset D$ is given by
$$ \lambda_t(X)=\int_X \mu^+_k(\x)d\x $$
In a Riemannian manifold with its curvature bounded as $0< k_1^2 \leq \kappa \leq k_2^2$, we have 
$$ \int_X \mu^+_{k_2}(\x)d\x \leq \lambda_t(X) \leq \int_X \mu^+_{k_1}(\x)d\x $$
%

\input{bibliostuff.tex}

\end{document}

%% file: bibliostuff.tex

\bibliographystyle{plain}
\bibliography{../../bookstore/coarse,../../bookstore/chaos,../../bookstore/oxford,../../bookstore/networking,../../bookstore/edmond,../../bookstore/biomed,../../bookstore/margareta,../../bookstore/bob}

%% file: JLBB.bbl
\begin{thebibliography}{10}

\bibitem{eurasip_clustering}
F.~Ariaei, M.~Lou, E.~Jonckeere, B.~Krishnamachari, and M.~Zuniga.
\newblock Curvature of sensor network: clustering coefficient.
\newblock {\em EURASIP Journal on Wireless Communications and Networking},
  2008.
\newblock to appear.

\bibitem{Berger2000}
Marcel Berger.
\newblock {\em {R}iemannian Geometry During the Second Half of the Twentieth
  Century}, volume~17 of {\em University Lecture Series}.
\newblock American Mathematical Society, Providence, RI, 2000.

\bibitem{bonk_schramm}
M.~Bonk and O.~Schramm.
\newblock Embeddings of {G}romov hyperbolic spaces.
\newblock {\em Geom. Funct. Analysis}, 10:266--306, 2000.

\bibitem{BridsonHaefliger1999}
Martin~R. Bridson and Andr\'e Haefliger.
\newblock {\em Metric Spaces of Non-Positive Curvature}, volume 319 of {\em A
  Series of Comprehensive Surveys in Mathematics}.
\newblock Springer, New York, NY, 1999.

\bibitem{Cisco05}
Cisco.
\newblock How does load balancing work?, 2005.
\newblock Document ID: 5212,2005.

\bibitem{mohar}
M.~DeVos and B.~Mohar.
\newblock An analogue of the {D}escartes-{E}uler formula for infinite graph and
  {H}iguchi's conjecture.
\newblock {\em Transactions of the American Mathematical Society},
  359:3275--3286, 2007.

\bibitem{docarmo}
M.~P. do~Carmo.
\newblock {\em Riemannian Geometry}.
\newblock Birkhauser, Boston, Basel, Berlin, 1992.

\bibitem{Helgason}
S.~Helgason.
\newblock {\em Differential Geometry and Symmetric Spaces}.
\newblock Chelsea Publishing, Providence, RI, 2000.

\bibitem{japanese_combinatorial_curvature}
Y.~Higuchi.
\newblock Combinatorial curvature for planar graphs.
\newblock {\em J. of Graph Theory}, 38:220--229, 2001.

\bibitem{4_point}
E.~Jonckheere, F.~Ariaei, and P.~Lohsoonthorn.
\newblock Upper bound on scaled hyperbolic gromov $\delta$: 4-point condition.
\newblock {\em Mathematics of Networks}, ?:?, 2009.
\newblock to be submitted.

\bibitem{scaled_gromov}
E.~Jonckheere, P.~Lohsoonthorn, and F.~Bonahon.
\newblock Scaled {G}romov hyperbolic graphs.
\newblock {\em Journal of Graph Theory}, 57:157--180, 2008.
\newblock DOI 10.1002/jgt.20275.

\bibitem{JonckheereLohsoonthornMED2002}
E.~A. Jonckheere and P.~Lohsoonthorn.
\newblock A hyperbolic geometry approach to multi-path routing.
\newblock In {\em Proceedings of the 10th Mediterranean Conference on Control
  and Automation (MED 2002)}, Lisbon, Portugal, July 2002.
\newblock FA5-1.

\bibitem{JonckheereLohsoonthornACC2004}
E.~A. Jonckheere and P.~Lohsoonthorn.
\newblock Geometry of network security.
\newblock In {\em Proceedings of the American Control Conference (ACC 2004)},
  Boston, MA, June 2004.
\newblock Paper WeM11.1, Special Session on Control and Estimation Methods in
  Network Security and Survivability.

\bibitem{Jost1997}
J.~Jost.
\newblock {\em Nonpositive Curvature: Geometric and Analytic Aspects}.
\newblock Lectures in Mathematics. Birkhauser, Basel-Boston-Berlin, 1997.

\bibitem{Jost1998}
J.~Jost.
\newblock {\em Riemannian Geometry and Geometric Analysis}.
\newblock Universitext. Springer, Berlin, Heidelberg, New York, 1998.
\newblock Second Edition.

\bibitem{KobayashiNomizu1996b}
S.~Kobayashi and K.~Nomizu.
\newblock {\em Foundation of Differential Geometry}, volume~2.
\newblock Wiley, 1996.

\bibitem{dima_hidden_hyperbolic}
D.~Krioukov, F.~Papadopoulos, A.~Vahdat, and M.~Boguna.
\newblock Curvature and temperature of complex networks.
\newblock {\em Physical Review E}, 80:035101(R), 2009.

\bibitem{Matt_thesis}
P.~Lohsoonthorn.
\newblock {\em Hyperbolic Geometry of Networks}.
\newblock PhD thesis, Department of Electrical Engineering--Systems, University
  of Southern California, 2003.
\newblock Available at http://eudoxus.usc.edu/IW/MATTFINALTHESIS\_MAIN.pdf.

\bibitem{mingjithesis}
M.~Lou.
\newblock {\em Traffic pattern analysis in negatively curved networks}.
\newblock PhD thesis, University of Southern California, Los Angeles, CA, May
  2008.
\newblock Available at http://eudoxus.usc.edu/IW/Mingji-PhD-Thesis.pdf.

\bibitem{arXiv_dmitri}
O.~Narayan and I.~Saniee.
\newblock The large scale curvature of networks.
\newblock arXiv:0907.1478v1, [cond-mat.stat-mech] 9 July 2009, July 2009.

\bibitem{reti_4_combinatorial_curvature}
T.~R\'eti, E.~Bitay, and Z.~Kosztol\'anyi.
\newblock On the polyhedral graphs with positive curvature.
\newblock {\em Acta Polytechnica Hungarica}, 2(2):19--37, 2005.

\bibitem{rocketfuel}
N.~Spring, R.~Mahajan, D.~Wetherall, and T.~Anderson.
\newblock Measuring {ISP} topologies with rocketfuel.
\newblock {\em {IEEE} Transactions on Networking}, 12(1):2--16, February 2004.

\bibitem{positively_curved_graphs}
Liang Sun and Xingxing Yu.
\newblock Positively curved cubic plane graphs are finite.
\newblock {\em Journal of Graph Theory}, 47(4):241--274, 2004.

\bibitem{congestion_tree}
L.~Zhao, Y.-C. Lai, K.~Park, and N.~Ye.
\newblock Onset of traffic congestion in complex networks.
\newblock {\em Physical Review E}, 71:026125--1--026125--8, 2005.

\end{thebibliography}
